\documentclass[10pt,a4paper,english]{article}
\usepackage{amsmath,amscd,amsfonts,eucal,latexsym,amssymb,mathrsfs}
\usepackage{amsxtra, amsthm, calc, bbm}
\newlength{\papersidemargin}
\newlength{\papertopmargin}
\newlength{\maxcarwidth}
\theoremstyle{definition}

\theoremstyle{plain}

\theoremstyle{remark}
\newtheorem*{Remark}{Remark}

\theoremstyle{definition}
\newtheorem{definition}{Definition}[section]
\theoremstyle{plain}
\newtheorem{theorem}[definition]{Theorem}
\newtheorem{proposition}[definition]{Proposition}
\newtheorem{lemma}[definition]{Lemma}

\newcounter{remcount}

\newcounter{propcount}

\newlength{\maxlabelwidth}

\def\cD{{\cal D}}

\def\cF{{\cal F}}

\def\cH{{\cal H}}

\def\cK{{\cal K}}

\def\bC{{\mathbb C}}

\def\bN{{\mathbb N}}

\def\bR{{\mathbb R}}

\def\bZ{{\mathbb Z}}
\def\a{\alpha}
\def\b{\beta}
\def\g{\gamma}        \def\G{\Gamma}
\def\d{\delta}

\def\e{\eta}

       \def\L{\Lambda}
\def\m{\mu}
\def\n{\nu}

\def\p{\pi}

\def\s{\sigma}

\def\t{\tau}

\def\o{\omega}        \def\O{\Omega}
\newcommand{\sD}{\mathscr{D}}

\newcommand{\sH}{\mathscr{H}}

\newcommand{\sJ}{\mathscr{J}}
\newcommand{\sK}{\mathscr{K}}
\newcommand{\sL}{\mathscr{L}}

\newcommand{\sP}{\mathscr{P}}

\newcommand{\sS}{\mathscr{S}}

\newcommand{\sU}{\mathscr{U}}

\renewcommand{\Im}{\textup{Im}\,}

\newcommand{\Lg}{\sL}

\newcommand{\bp}{\boldsymbol{p}}

\newcommand{\onp}{\o_{\bar{n}}(\bp)}
\newcommand{\orp}{\o_r(\bp)}
\newcommand{\ocp}{\o_c(\bp)}

\newcommand{\wick}[1]{:\mspace{-3mu}#1\mspace{-7mu}:\mspace{-3mu}}

\newcommand{\inst}[1]{$^\textrm{#1}$ }
\newcommand{\email}[1]{e-mail: #1}

\begin{document}
\title{Local aspects of free open bosonic\\
 string field theory}
\author{Luca Tomassini\inst{1}
}

\date{
\parbox[t]{0.9\textwidth}{\footnotesize{%
\begin{enumerate}
\renewcommand{\theenumi}{\arabic{enumi}}
\renewcommand{\labelenumi}{\theenumi}
\item Dipartimento di Matematica, Universit\`a di Roma ``Tor Vergata'', Via della Ricerca Scientifica, 1, I-00133 Roma, Italy, \email{tomassin@mat.uniroma2.it} 
\end{enumerate}
}}
\\
\vspace{\baselineskip}
\today}

\maketitle

\begin{abstract}
We show that strictly local observables with arbitrarily small support in space-time exist in covariant free bosonic string field theory. The main ingredient of the proof is a modified version of the well known DDF operators, which we rigourously define. This result allows in principle the definition of a net of local observable algebras and should be considered as a first step towards the application of methods of algebraic quantum field theory to this case.
\end{abstract}

\section{Introduction}\label{sec:intro}
Interest for an axiomatic approach to string field theory has recently been revived by the work of Dimock, which in a series of papers set up the ground for a rigourous treatment of canonically quantized free bosonic string fields both in the light-cone \cite{Dimock2000} and covariant \cite{Dimock2002} case. The aim of those papers was to put on a firmer ground the study of one important and only rarely touched upon question: the localization properties of the string fields. In fact, the string being an extended object, these are commonly expected to be substantially different from the point particle case. At present, the results of Martinec \cite{Martinec1993} and Hata and Oda \cite{Hata1997} (but see also \cite{Dimock2000}) are widely accepted as a satisfactory answer to this problem and are formulated as follows. The string field is considered as a function of parametrized strings $X=X^\m (\s)$ and its commutator $[\Phi(X),\Phi(Y)]$ is found to vanish if
\begin{equation}\nonumber
\int_0^\pi (X(\s)-Y(\s))^2d\s >0
\end{equation}
(here and in the following we make use of the usual ``string theoretic'' signature of the Minkowski metric), a condition which is often considered as unintuitive since it does not require that all points on the two strings be space-like separated. Notice that this is different from
\begin{equation}\nonumber
\left( \int_0^\pi X(\s)d\s-\int_0^\pi Y(\s)d\s\right)^2 >0
\end{equation}
expressing space-like separation of the respective centers of mass. Unfortunately these formulas can be considered meaningful only in the case of light-cone quantization, where one deals only with {\it observable} quantities. In the covariant approach this is no longer the case and it is a fact that after quantization neither $X(\s)$ nor $\int_0^\pi X(\s)d\s$ are invariant under reparametrizations of the world-sheet.\\
A more intrinsic point of view on the problem is thus in order, in which the meaning attached to the word ``locality'' does not depend too much on a physical interpretation of the $X(\s)$'s. This was started in \cite{Dimock2000} and \cite{Dimock2002}, where the string fields were found to have perfectly local commutators (in the ordinary sense) both in the light-cone and covariant quantization, but in this second case unobservable quantities and an indefinite metric state space have to be introduced \cite{Strocchi1993}. As for the free electromagnetic field, the physical theory is then recovered by a Gupta-Bleuler-like procedure and in particular observable quantities can be obtained from the fields by smearing with test functions satisfying certain conditions. In case of incompatibility of these conditions with compact support, the string fields would be local but at the same time no local {\it observables} would exist. In the language of algebraic quantum field theory \cite{Haag:1996a}, this is the situation of a local field net having no local subnet left fixed by the gauge group action.\\
The main result of this paper is a proof that this is not the case: local observables with arbitralily small compact support in space-time do exist in free open bosonic string field theory (string field theory from now on), suggesting that either this construction is not suitable to capture the extended nature of strings or at least that this characteristic has little consequence on their local behaviour.\\
The free string field being obtained by a more or less ordinary (indefinite metric) Fock construction, in section \ref{first} we give a proper formulation of first covariant {\it canonical} quantization of the open bosonic string with flat (Minkowsky) target space (simply string in the following) and of its very basic structures (compare with the standard textbook \cite{Green1987})\footnote{There are other important approaches to the quantization of strings. Besides the geometric ones (see for example \cite{Mickelsson1987},\cite{Bowick1987},\cite{Bowick1987a}), we mention the algebraic approach of Pohlmeyer and Rehren in which only gauge invariant quantities are quantized and no critical dimension appears \cite{Polmeyer1986},\cite{Polmeyer1988},\cite{Meusburger2003}. In \cite{Bahns2004} it was proved that this procedure is $not$ equivalent to canonical quantization. However, we do not regard canonical quantization (and indeed quantization in general) as a physically meaningful procedure but rather as a tool to construct reasonable quantum models. We also hope that our attempt to put the use of DDF operators on a firmer mathematical ground will help a better understanding of their relations with the Polmeyer invariants (see \cite{Schreiber}).}. While this is certainly not the first attempt in this direction (see for example \cite{Grundling1993}), one important ingredient of the standard canonical approach was still missing: DDF operators (from the names of Del Giudice, Di Vecchia and Fubini \cite{DelGiudice1972}). They are gauge invariant operators (in the sense that they commute with the constraints) and as such have proved to be a powerful tool in the construction of physical states. Acually, the very first proof of the no-ghost theorem (\cite{Brower1972},\cite{Goddard1972}) was based on this property. In subsection \ref{sub:ddf}, after dwelling in a somewhat detailed study of their mathematical properties, we follow this strategy.\\
In section \ref{field} we define and study the string field, making use of the above mentioned results on the first quantized theory to construct physical test functions (and hence observables) with arbitrarily small support.\\

\section{First quantized covariant strings}\label{first}

In this section we collect some known facts about canonical quantization of free open bosonic strings with flat Monkowski target space (simply string in the following). We follow the classical approach of \cite{Green1987}, as reformulated in \cite{Dimock2002} and \cite{Grigore2007}. The material is more or less standard, the only exception being subsection \ref{sub:ddf}. There, the transverse DDF operators are introduced and used to construct physical string states (wave functions) with particular regularity properties. We also hope to help filling what is to our knowledge a gap in the literature, where a rigourous definition of DDF operators seems to be missing (but see \cite{Grigore2007} for work in this direction).\\

\subsection{General structures}
Consider the Hilbert space $\sL^2(\bR^d)$ and the multiplication operators $p^\m$ ($\m = 0,\dots,d-1$) on it as defined on $\sS(\bR^d)$ (the Schwartz space over $\bR^d$). It is essentially self-adjoint there, but in the following we will distinguish operators and their closures only if needed. On the same domain, we can also define ``position'' operators $x_\m=-i\partial/\partial p^\m$. Moreover, consider $x^\m=\sum_\n \e^{\m\n}p_{\n}$ with the metric tensor $\e=(-1,\underbrace{1,...,1}_{d-1 \; \text{times}})$. Then on $\sS(\bR^d)$
\begin{equation}
[x^\m,p^\n]=i\e^{\m\n}
\end{equation}
For future convenience, we recall there is also an associated space-time representation, obtained by Fourier transformation. In the context of string theory, the operators $p^\m$ and $x^\m$ are costumarily identified with the string center of mass momentum and position.
\begin{Remark}
We emphasize that the operators $x^\m$ will have no role in the following: they only appear as the zero component of the Fourier decomposition of the string operator $X(\s,\t)$ (see equation \ref{stringoperators} below). It will actually be clear that both of them do $not$ commute with the constraints, so that at the quantum level they should not be considered as particularly relevant: only observables should.
\end{Remark}
\bigskip
To model out string oscillators one starts with the algebraic Fock space ${\cal F}_0$ of finite linear combinations of vector $\a^{\m_1}_{m_{1}}\dots \a^{\m_k}_{m_{k}}\Omega$, where the family of operators $\alpha^{\m}_{m}, m \in \bZ/\{0\}$ and the vector $\O$ are such that
\begin{eqnarray}\nonumber
[ \alpha^{\m}_{m}, \alpha^{\n}_{n} ] = m~\eta^{\m\n}\delta_{m+n} \cdot I& 
& \quad \forall m, n \neq 0
\nonumber \\
\alpha^{\m}_{m} \Omega = 0\qquad\quad& &\quad m > 0
\nonumber \\
(\alpha^{\m}_{m})^{\dagger} = \alpha^{\m}_{-m}\qquad& &\quad \forall m\neq 0\nonumber
\end{eqnarray}
where $\d$ is the Kronecker delta.\\
On the space ${\cal F}_0$ there are a unique indefinite hermitian form $\langle\cdot,\cdot\rangle$ such that $\langle\O,\O\rangle=1$ and $\langle\a_n^\m\phi,\psi\rangle=\langle\phi,\a_{-n}^\m\psi\rangle$ and
an associated scalar product $(\cdot,\cdot)$, given by the the formula $\langle\cdot,\cdot\rangle=(\cdot,J\cdot)$ with $J\a^{\m_1}_{m_{1}}\dots \a^{\m_k}_{m_{k}}\Omega=\e^{\m_1\m_1}\a^{\m_1}_{m_{1}}\dots \e^{\m_k\m_k}\a^{\m_k}_{m_{k}}\Omega$. Since $J^2=I$, ${\cal F}_0$ has the structure of a pre-Krein space, {\it i.e.} the topologies generated by $\langle\cdot,\cdot\rangle$ and $(\cdot,\cdot)$ are equivalent \cite{Bognar1974} and we can consider its closure $\cF$ without specifications. As usual, we denote by $A^*$ the adjoint of an operator $A$ on ${\cal F}$ with respect to this scalar product, as opposed to the symbol $A^\dagger$ by which we refer to $J$-adjoints. Thus $\langle \phi,A\psi\rangle=\langle A^\dagger \phi,\psi\rangle$, $(\phi,A\psi)=(A^* \phi,\psi)$ and $A^\dagger =J A^* J$. We stress that $J=J^\dagger=J^*$.\\
Since $\forall \psi\in \cF_0$ we have $\alpha^{\m}_{n}\psi=0$ for sufficiently large $n>0$, on the subspace $\cF_0$ we can define
\begin{equation}
\tilde{L}_{m} = \frac{1}{2} \sum_{n \in \bZ/{0}}~\eta_{\m\n}\wick{\alpha^{\m}_{m-n} \alpha^{\n}_{n}}\;\;
= \frac{1}{2} \sum_{n \in \bZ/{0}} \wick{\alpha_{m-n}\cdot \alpha_{n}}
\end{equation}
($\wick{\cdot\;}$ indicates Wick ordering) and we have $\tilde{L}_{m}\O=0$ for $m>0$. Thus, the operators $\tilde{L}_{m}$ leave this domain invariant and then simple algebra is enough to check that they satisy
\begin{equation}
[ \tilde{L}_{m}, \tilde{L}_{n} ] = (m - n) \tilde{L}_{m+n} 
+ d\frac{ m (m^{2} - 1)}{12}~\delta_{m+n}~\cdot I
\end{equation}
so that they define a lowest weight representation of the Virasoro algebra with $(c,h)=(d,0)$.
Moreover, they are closable since their adjoints $\tilde{L}^*_m =J\tilde{L}_{-m}J$ are densely defined.\\ 
Finally, a representation $u$ of the (full) $d$-dimensional Lorentz group $\Lg$ can be defined on $\cF_0$ by
\begin{equation}\label{poi}
u(\L)\a^{\m_1}_{m_{1}}\dots \a^{\m_k}_{m_{k}}\Omega=\L^{\m_1}_{\n_1}\a^{\n_1}_{m_{1}}\dots \L^{\m_k}_{\n_k}\a^{\n_k}_{m_{k}}\Omega
\end{equation}
with $\L\in\Lg$ (so that $\cF_0$ is left invariant) and we have
\begin{equation}
u(\L)^{-1}\a^{\m}_{m}u(\L)=\L^{\m}_{\n}\a^{\n}_{m}
\end{equation}
The operators $u$ leave the indefinite inner product invariant but are in general not bounded.\\
Consider now the Krein space
\begin{equation}
\cK=\sL^2(\bR^d,dp)\otimes\cF=\sL^2(\bR^d,\cF,dp)
\end{equation}
We write the induced inner product as
\begin{equation}
\langle\phi,\psi\rangle=\int_{\bR^d}\langle\phi(p),\psi(p)\rangle dp\;\;\;\;\phi,\psi\in \cK
\end{equation}
where by a slight abuse we do not distiguish the notations for the inner product on the whole space $\cK$ and the one on the single fiber $\cK_p\sim \cF$ ($p\in \bR^d$). The associated scalar product is
\begin{equation}\label{scalar}
(\phi,\psi)=\int_{\bR^d}(\phi(p),\psi(p)) dp \;\;\;\;\phi,\psi\in \cK
\end{equation}
and in the following we will indicate by $\|\cdot\|$ both the corresponding norms on $\cK$ and $\cK_p\sim \cF$ ($p\in \bR^d$). 
The two products are related by $\langle\psi,\phi\rangle=(\phi,I\otimes J\psi)$ ($\phi,\psi\in\cK$), but from now on we will write $J$ for $I\otimes J$. Also, we will not change the notation concerning adjoints.\\
On the dense domain $\sS(\bR^d)\otimes\cF_0$ we now define operators $p^\m\otimes I$, $p^\m\otimes I$ and  $I\otimes\a_n^\m$  ($n\in \bZ/\{0\}$), but continue to use the symbols $p^\m$, $x^\m$ and $\a_n^\m$ respectively.
Then, the string operators are given by
\begin{equation}\label{stringoperators}
X^\m(\s,\t)=x^\m + p_0^\m\t +i\!\!\sum_{n\in \bZ/\{0\}}\a_n^\m e^{-in\t}\frac{\cos{n\s}}{n}
\end{equation}
for $(\s,\t)\in [ 0,2\p]\times\bR$. They are operator valued distributions in the variable $\s$ and are formally satisfy $X^\m(\s,\t)=X^\m(\s,\t)^\dagger$.\\
Introducing the notation $\a_0^\m=p^\m$, the constraint operators are defined on the same domain by the formula
\begin{equation}
L_m=\frac{1}{2}\sum_{n\in\bZ}\wick{\a_{m-n}\cdot\a_n}-b\d_mI
\end{equation}
with $b\in\bR$ (we introduce the standard string theoretic modification of the operator $L_0$) and satisfy
\begin{eqnarray}
L^{\dagger}_m&=&L_{-m}\\
\left[ L_{m}, L_{n} \right] = (m - n) L_{m+n} 
\!&+&\! \begin{bmatrix}d\frac{ m (m^{2} - 1)}{12}+2bm\end{bmatrix}\delta_{m+n}~\cdot I
\end{eqnarray}
for all $m\in\bZ$.
We note for future convenience that, introducing the operators
\begin{equation}\label{N}
N=\frac{1}{2} \sum_{n\in\bZ/\{0\}}\wick{\a_{-n}\cdot\a_n}
\end{equation}
and $M^2=2(N-a)$ (as defined on $\sS(\bR^d)\otimes\cF_0$), we can write $2L_0=(p^2 +M^2)$. In the space-time representation this reads $-\Box +M^2$, the Klein Gordon operator for $\cF$-valued functions. Observing that $JNJ=N$ (so that adjoints and $J$-adjoints coincide for $N$) and recalling that the operators $\tilde{L}_m$ ($m\in\bZ$) are closable, a slight modification of the aguments of (\cite[Lemma 5]{Dimock2002}) gives the following
\begin{lemma}\label{lem:mass}
The operators $L_m$ ($m\in\bZ$) are closed on the domain
\begin{multline}
D(L_m)= \{ \psi\in \cK : \psi(p)\in D(L_m(p))\quad \text{a.e.}\; p,\\ \int_{\bR^d}\|L_m(p)\psi(p)\|^2dp <\infty\}
\end{multline}
The operators $N$ and $M^2$ are self-adjoints on the domain
\begin{equation}
D(N)= \{ \psi\in \cK : \psi(p)\in D(N) \quad \text{a.e.}\; p, \int_{\bR^d}\|N\psi(p)\|^2dp <\infty \}
\end{equation}
and their spectrum is discrete with finite multiplicity. In particular, the spectrum of $M^2$ is $\{2(n-b), n\geq0\}$.
\end{lemma}
\bigskip
Finally, there is a representation $U$ of the full Poincar\'e group $\sP$ defined on $\sS(\bR^d)\otimes\cF_0$ by
\begin{equation}
(U(a,\L)\psi)(p)=e^{-ip\cdot a}u(\L)\psi(\L p)
\end{equation}
with $a\in\bR^d$ and $\L\in\Lg$.
This preserves the indefinite inner product by the Lorentz invariance of Lebesgue measure. Notice that, for $n\in\bZ$ and on $\sS(\bR^d)\otimes\cF_0$,
\begin{equation}\begin{split}\nonumber
U(a,\L)^{-1}x^\m U(a,\L)=x^\m+a^\m\\
U(a,\L)^{-1}\a^\m_n U(a,\L)=\L_\n^\m\a_n^\n
\end{split}\end{equation}
It follows that on the same domain
\begin{equation}\begin{split}\nonumber
U(a,\L)^{-1}L_mU(a,\L)=L_m \qquad\qquad\qquad \forall m\in\bZ\\
U(a,\L)^{-1}X^\m(\s,\t)U(a,\L)=\L^\m_\n X^\n(\s,\t) + a^\m 
\end{split}\end{equation}
\subsection{Imposing the constraints. The physical theory}\label{sub:constraints}
As it is well known \cite{Green1987}, invariance of the quantum string with respect to world-sheet reparametrizations is imposed in the weak form
\begin{eqnarray}
L_m\psi=0 \;\;\;\; m\geq0
\end{eqnarray}
and for example, according to the discussion after equation (\ref{N}), the constraint $L_0\psi=0$ is equivalent 
to the equation $(p^2 +M^2)\psi=0$. Not even this condition has any chance to be met because of the discreteness of the spectrum of $M^2$. As a consequence, the Hilbert space has to be ``reconfigured'' so that this becomes possible: once more, we follow \cite{Dimock2002}.
There, it was proven that setting
\begin{equation}\nonumber
V_r=\{ p\in\bR^d/\{0\} : p^2+r=0\} 
\end{equation}
it holds (as Hilbert spaces)
\begin{equation}
\cK=\sL^2(\bR^d,\cF, dp)= \int^{\otimes}\sL^2(V_r,\cF, d\m_r(\bp))dr=\int^{\otimes}\cK_r dr
\end{equation}
where $\m_r(\bp)$ indicates the Lorentz invariant measure on it and we put $\cK_r=\sL^2(V_r,\cF, d\m_r(\bp))$.
This means (see for example \cite{Naimark1972})
\begin{equation}
\int_{\bR^d}\|\psi(p)\|dp=\int_{-\infty}^{\infty} \left(\int_{V_r}\|\psi(\orp,\bp)\|^2d\m_r(\bp)\right) dr
\end{equation}
where identification of $\psi\in\cK$ with $r\rightarrow\psi_r$ is understood ($\psi_r$ is the restriction of $\psi$ to $V_r$). Since
\begin{equation}
\langle\psi,\phi\rangle=\int_{-\infty}^{\infty}\langle\psi_r,\phi_r\rangle dr
\end{equation}
the previous decomposition is also a decomposition of Krein spaces. Correspondingly, all the operators considered so far can be written as direct integrals of their restrictions to the $V_r$'s.\\
Indicating with  $V^+_r$ the $p^0>0$ half of $V_r$, we make the following
\begin{definition}\label{singlestring}
\begin{enumerate}
\item
The Krein space for the single string is 
\begin{equation}\begin{split}
\cH  =&  \bigoplus_{r= -2,0,2,... } \ \   \cK_r  \\
 \cK_r =& \sL^2(V^{(+)}_r, \cF,  d\m_r(\bp)) \\
\end{split}
\end{equation}
where   $  V^{(+)}_r = V_r^+$ for  $r \geq 0$ and   $  V^{(+)}_r = V_r$ for  $r < 0$ .
\item For   
$\psi, \chi \in  \cH$ with restrictions $\psi_r, \chi_r$, an indefinite inner product is defined by  
\begin{equation}  
\langle\psi , \chi\rangle = \sum_{r= -2,0,2,... }  \langle\psi_r, \chi_r\rangle     
\end{equation} 
\item A representation of the Poincar\'e group $\sP$
is defined by 
\begin{equation}
U(a,\L) =  \bigoplus_{r= -2,0,2,... }  U_r(a, \L)  
\end{equation}
\item The constraint operators are defined by   
\begin{equation}
\begin{split}
L_0 = & \bigoplus_{r= -2,0,2,... }  \frac12 (-r + M^2) \\
L_m  = & \bigoplus_{r= -2,0,2,... } L_{r,m} \\
\end{split}
\end{equation}
\end{enumerate}
\end{definition}
\bigskip
As expected, the operators $L_m$ are closed on the domain
\begin{multline}
D(L_m)= \{ \psi\in \cH : \psi(p)\in D(L_m(p)) \quad a.e. \; p,\\ \sum_{r= -2,0,2,... }\int_{\bR^d}\|L_m(\orp,\bp)\psi(\orp,\bp)\|^2d\m_r(\bp) <\infty\}
\end{multline}
while $N$ (and hence $M^2$) are selfadjoint and $J$-selfadjoint on
\begin{multline}
D(N)= \{ \psi\in \cH : \psi(p)\in D(N)\quad a.e. \; p,\\ \sum_{r= -2,0,2,... }\int_{\bR^d}\|N\psi(\orp,\bp)\|^2d\m_r(\bp) <\infty\}
\end{multline}
With this preparation, we can now define the physical state space of the string. First, set
\begin{equation}\begin{split}
\cH^{\prime}&=\sum_{r= -2,0,2,... }\cK_r^{\prime}\\
\cK_r^{\prime}=\{\psi\in\cK_r^{\prime} &: (-r+M^2)\psi=0,\quad L_m\psi=0 \quad m>0\}
\end{split}
\end{equation}
The subspace of isotropic (or spurious) elements of $\cH^{\prime}$ is defined by 
\begin{equation}\begin{split}
\cH^{\prime\prime}=\cH^{\prime}\cap&(\cH^{\prime})^{\perp}=\sum_{r= -2,0,2,... }\cK_r^{\prime\prime}\\
\cK_r^{\prime\prime}=&\cK_r^{\prime}\cap(\cK_r^{\prime})^{\perp}
\end{split}
\end{equation}
where orthogonality is in terms of the indefinite inner product. Then, if $\psi\in\cH^{\prime\prime}$ one has $\langle\psi,\psi\rangle=0$. The physical Hilbert space is defined as
\begin{equation}\begin{split}
\cH^{phys}=\cH^{\prime}/\cH^{\prime\prime}&=\sum_{r= -2,0,2,... }\cH_r^{phys}\\
\cH_r^{phys}=&\;\cK_r^{\prime}/\cK_r^{\prime\prime}
\end{split}
\end{equation}
The fact that the indefinite inner product restricts to a positive definite one on $\cH^{phys}$ when $d=26$ and $b=1$ is the content of the famuous no-ghost theorem of Brower \cite{Brower1972}, Goddard and Thorn \cite{Goddard1972}, as reformulated in \cite{Dimock2002}.

\subsection{Physical states and DDF operators}\label{sub:ddf}

The canonical quantization of the open string has been carried out in some detail in the previous subsection. However, in view of passing to second quantization ({\it i.e.} to string field theory) a closer look to the first quantized state space and especially to observable states is in order (see the next section). As already anticipated, our main tool will be the DDF operators introduced in \cite{DelGiudice1972} as reworked in \cite{Grigore2007}. Even in this last reference, though, a complete treatment of the string center of mass position and momentum is lacking and (modified) DDF operators are defined only for fixed strictly positive fixed masses (with the notation of subsection \ref{sub:constraints}, this means taking a positive fixed value of the mass parameter $r$). In what follows we will introduce DDF operators on the whole $\cK=\sL^2(\bR^d,\cF, dp)$. The definition of their restrictions to $\cH$ can easily be obtained as in subsection \ref{sub:constraints}.\\
Let us begin by observing that the expressions
\begin{equation}\begin{split}
U_{0}(k)&=I \\
U_{n}(k) = \sum_{p \geq 0} \frac{1}{p!} \sum_{n_{1},\dots,n_{p} > 0}
\sum_{n_{1}+\cdots +n_{p}=n}~&\frac{1}{ n_{1} \dots n_{p}} 
(k\cdot \alpha_{n_{1}}) \cdots (k\cdot \alpha_{n_{p}}) \;\;\;n>0
\end{split}
\end{equation}
are well defined on the algebraic Fock space $\cF_0\subset \cF$ and leave it invariant. Moreover, we see that 
for $\psi \in {\cal F}_{0}$
\begin{equation}\label{finite1}
U_{n}(k) \psi = 0
\end{equation}
for $n$ sufficiently large. Then, the same properties hold for
\begin{equation}\label{finite2}\begin{split}
V_{n}(k) =& \sum_{p \in \bZ} U_{p-n}(-k)^{\dagger}~U_{p}(k)  \\
\bar{V}^{\mu}_{n}(k) = \sum_{p > 0} &[ \alpha^{\mu}_{-p} V_{n+p}(k) + V_{n-p}(k) \alpha^{\mu}_{p} ]
\end{split}\end{equation}
with $n\in\bZ$.\\
We now go over to the space $\cK=\sL^2(\bR^d,\cF, dp)$ and define on the invariant domain $\sS(\bR^d)\otimes\cF_0$ the operators
\begin{equation}\label{finite3}
V^{\mu}_{n}(k) = \bar{V}^{\mu}_{n}(k) + p^{\mu}~V_{n}(k)
\end{equation}
Then on the same domain, under the hypothesis that $k^2=0$ ({\it i.e.} the vector $k\in\bR^d$ is light-like), the following relations hold true:
\begin{align}\label{ddf1}
[ V^{\mu}_{m}(mk), V^{\n}_{n}(nk) ] = - \eta^{\mu\nu}&m~\delta_{m+n}\cdot I
+ k^{\mu}~V^{\nu}_{m,n}(k) - k^{\nu}~V^{\mu}_{n,m}(k)\nonumber\\
[ L_{0}, V^{\mu}_{n}&(nk) ] = -nV^{\mu}_{n}(nk)\\
[ L_{m}, V^{\mu}_{n}(nk) ] = - n~(1 + k\cdot p)&V^{\mu}_{m+n}(nk)
+ \frac{m(m-1)}{2}nk^{\mu}V_{m+n}(nk)\quad m\neq0\nonumber
\end{align}
where the explicit expressions of the terms $V^{\nu}_{m,n}(k)$ are not important since it is clear they leave $\sS(\bR^d)\otimes\cF_0$ invariant.\\
Observe now that if $k^i=0$ for $i=1,\cdots,d-2$ and $k\cdot p=-1$ the relations (\ref{ddf1}) greatly simplify. This motivates introducing, on the domain
\begin{equation}
\cD_0^{p^+}=\{f\in\sS(\bR^d) : \lim_{(p^0+p^{d-1})\to0}(p^0+p^{d-1})^{-\g} f(p)=0 \;\;\text{for any}\;\; \g\in\bN \},
\end{equation}
the operators $k^\m(p)$ ($\m=0,\cdots,d$) given by
\begin{equation}\label{cappa}\begin{split}
k^\m(p)=0\quad\quad\quad\quad \;\;\;\;\;&\quad p\in\bR^d, \quad\m=1\cdots,d-2\\
k^0(p)=-k^{d-1}(p)=\frac{1}{2(p^0+p^{d-1})}\;\;\;\;\;&\quad p\in\bR^d,\quad p^0+p^{d-1}\neq0\\
k^0(p)=k^{d-1}(p)=0\;\;\;\;\;\qquad\quad&\quad p\in\bR^d,\quad p^0+p^{d-1}=0
\end{split}\end{equation}
We are now ready for the following
\begin{definition}
The operators $A^i_n$ ($i=1\cdots,d-2$) on $\cK=\sL^2(\bR^d,\cF, dp)=\sL^2(\bR^d,dp)\otimes\cF$ given on each fiber $\cK_p\sim\cF$ ($p\in\bR^d$) by
\begin{equation}\begin{split}
A^i_n(p)&=V^i_n(nk(p)) \;\;\;\;\;\;\text{if}\quad k(p)\neq0\\
&A^i_n(p)=0 \;\;\;\;\,\quad\quad\text{if}\quad k(p)=0
\end{split}\end{equation}
are called transverse DDF operators.
\end{definition}
\bigskip
We collect their main properties in the following
\begin{theorem}\label{thm:ddf}
The transverse DDF operators $A^i_n$ ($i=1\cdots,d-2$) are are well defined on the invariant domain $\cD_0^{p^+}\otimes\cF_0$ and are closed operators on the domain
\begin{equation}\nonumber
D(A^i_n)= \{ \psi\in \cK : \psi(p)\in D(A^i_n(p))\; a.e.\; p, \int_{\bR^d}\|A^i_n(p)\psi(p)\|^2dp <\infty\}
\end{equation}
Moreover, since $L_{m}(\cD_0^{p^+}\otimes\cF_0)\subset\cD_0^{p^+}\otimes\cF_0$ for every $m\in\bZ$, there holds
\begin{align}\label{ddf2}
[ A^{i}_{n}, A^{j}_{m} ] &= \delta_{ij}~m~\delta_{m+n}\cdot I
\nonumber \\
~[ L_{m}, A^{i}_{n} ] &= 0, \quad\quad\quad\qquad\qquad\qquad\qquad m \not= 0 \quad \nonumber \\
~[ L_{0}, A^{i}_{n} ] &= - n~A^{i}_{n}\quad\quad\quad\\
(A_{n}^{i})^{\dagger} &= A^{i}_{-n}\quad\quad\quad\quad
\nonumber \\
A^{i}_{n} \Omega &= 0 \quad\quad\quad\qquad\qquad\qquad\qquad\; n > 0\quad
\nonumber \\
A^{i}_{0} &= p^{i}\quad\quad\quad\quad\quad\nonumber
\end{align}
\end{theorem}
\begin{proof}
First, observe from equations (\ref{finite1}),(\ref{finite2}) and (\ref{finite3}) that application of any of the $A^i_n$'s to vectors $f\otimes\chi\in\cD_0^{p^+}\otimes\cF_0$ results in a {\it finite} sum of terms of the form
\begin{gather}\nonumber
K\cdot [(p^{i_{1}}\cdots p^{i_{l}})\frac{1}{(p^0+p^{d-1})^\g} (k(p)\cdot \alpha_{-n_{1}}) \cdots\\ 
\nonumber \cdots(k(p)\cdot \alpha_{-n_{q}})\times(k(p)\cdot \alpha_{n_{1}^\prime}) \cdots (k(p)\cdot \alpha_{n_{s}^\prime})](f(p)\otimes \chi)=\\
\nonumber =K\cdot [ (p^{i_{1}}\cdots p^{i_{l}})\frac{1}{(p^0+p^{d-1})^{\g^\prime}} (\alpha_{-n_{1}}^0 +\alpha_{-n_{1}}^d) \cdots \\
\nonumber \cdots(\alpha_{-n_{q}}^0 +\alpha_{-n_{q}}^d)\times(\alpha_{n_{1}}^{0}+\alpha_{n_{1}}^d) \cdots (\alpha_{n_{s}}^0 +\alpha_{n_{s}}^d)](f(p)\otimes \chi)=\\
=K\cdot [(p^{n_{1}}\cdots p^{n_{l}}) \frac{1}{(p^0+p^{d-1})^{\g^\prime}} f(p)]\otimes [(\alpha_{-n_{1}}^0 +\alpha_{-n_{1}}^d) \cdots\\
\nonumber \cdots(\alpha_{-n_{q}}^0 +\alpha_{-n_{q}}^d)\times (\alpha_{n_{1}^\prime}^{0}+\alpha_{n_{1}^\prime}^d) \cdots (\alpha_{n_{s}^\prime}^0 +\alpha_{n_{s}^\prime}^d)\chi]
\end{gather}
where $l,q,s,\g,\g^\prime\in\bN$ and $K\in\bC$ are constants depending on $\chi\in\cF_0$ and the $A^i_n$'s, $n^1,\cdots,n^q,n^\prime_1,\cdots,n^\prime_s\in\bN$ and $i_1,\cdots,i_l=1,\cdots,d-2$. Since $f\in\cD_0^{p^+}$, it is then clear that this is well defined and that $\cD_0^{p^+}\otimes\cF_0$ is invariant. Moreover, $\cD_0^{p^+}\otimes\cF_0$ is dense (this last statement follows for example from density of $\sS(\bR^d)$ in $\sL^2(\bR^d)$ and absolute continuity of the Lebesgue measure \cite{Reed1980}).\\
Closure of $A^i_n$ on $D(A^i_n)$ can be proved as follows. To begin with, take $\psi_l\in D(A^i_n)$ so that $\|\psi_l-\psi\|\to0$ and $\|A^i_n\psi_l-\phi\|\to0$ and observe that by the completeness of $\sL^2(\bR^d,dp)\otimes\cF$ there is a subsequence $\psi_{l_k}$ such that almost everywhere we have $\|\psi_{l_k}(p)-\psi(p)\|\to0$ and $\|A^i_n(p)\psi_{l_k}(p)-\phi(p)\|\to0$. Since $A^i_n(p)^*=JA^i_{-n}(p)J$ is densely defined, the operator $A^i_n(p)$ is closable and $\phi(p)\in D(A^i_n(p))$ with $A^i_n(p)\psi(p)=\phi(p)$ almost everywhere.\\
The commutation relations in (\ref{ddf2}) follow from the ones in (\ref{ddf1}) and the fact that $\cD_0^{p^+}\otimes\cF_0$ is invariant for the operators $L_m$.
\end{proof}
\bigskip

\begin{Remark}
\begin{enumerate}
\item
We observe that the domain $\cD_0^{p^+}\otimes\cF_0$ is $not$ left invariant by the representation $U(a,\L)$ of the Poincar\'e group, nor is $D(A^i_n)$.
\item
The transverse DDF operators are the $z$-independent component of the ``vertex'' operator
\begin{equation} 
\dot{\widetilde{X}^{i}}(z) e^{ik\cdot \widetilde{X}(z,nk)}
\end{equation}
where $i=1,\cdots,d-2$ and
\begin{equation}
\widetilde{X}^{\mu}(z) = \sum_{n \neq 0} \frac{1}{n} \alpha^{\mu}_{n}~z^{n} + p^{\mu}~ln(z)
\qquad
\dot{\widetilde{X}^{i}}(z) = \sum_{n \neq 0} \alpha^{i}_{n}~z^{n} + p^{i}
\end{equation}
This is fundamentally different from the classical definition (see \cite{Green1987} but compare also with \cite{Buchholz1988}), since our ``modified'' DDF operators leave the single fiber $\cK_p$ of $\cK$ invariant and do not commute with $L_0$. This last fact also implies that DDF states ({\it i.e.} states of the form $\psi=A^{i_1}_{m_{1}}\dots A^{i_k}_{m_{k}}(f\otimes\Omega)$) do not necessarily satisfy the Klein-Gordon constraint $L_0\psi=0$ neither on $\cK$ nor on the reconfigured $\cH$ even if $f\otimes\Omega$ does. We will see in the next section that this is actually not a problem in the derivation of physical test functions for the string field.
\end{enumerate}
\end{Remark}
\bigskip
In particular, we have for $f\in\cD_0^{p^+}$
\begin{equation}\label{constraintsf}\begin{split}
L_mA^{i_1}_{-n_{1}}\dots A^{i_k}_{-n_{k}}&(f\otimes\Omega)=0\qquad m>0\\
L_0A^{i_1}_{-n_{1}}\dots A^{i_k}_{-n_{k}}(f\otimes\Omega)=&\frac{1}{2}(p^2-2b+2\sum_{l=1}^{k}n_l)(f\otimes\Omega)
\end{split}\end{equation}
where $n_1,\cdots, n_k >0$. Moreover, for future convenience we observe that together with the equality $M^2=2(N-b)$ this implies
\begin{equation}\label{salvezza1}
M^2A^{i_1}_{-n_{1}}\dots A^{i_k}_{-n_{k}}(f\otimes\Omega)=2(-b+\sum_{l=1}^{k}n_l)(f\otimes\Omega)
\end{equation}
so that, if $P_r$ indicates the projection on the subspace $M^2=r=-2,0,2,\cdots$ and $\bar{n}=\sum_{l=1}^{k}n_l$, we have
\begin{equation}\label{salvezza2}
P_rA^{i_1}_{-n_{1}}\dots A^{i_k}_{-n_{k}}(f\otimes\Omega)=\d_{r,2(-b+\bar{n})}A^{i_1}_{-n_{1}}\dots A^{i_k}_{-n_{k}}(f\otimes\Omega)
\end{equation}
We now want to show that the domain $D(A^i_n)$ always contains functions with suitable regularity properties.
From the proof of theorem \ref{thm:ddf} it is clear that if $g$ is such that there is a large enough positive integer, say $\tilde{\g}$, with $\lim_{p^0+p^1\to0}(p^0+p^1)^{-\tilde{\g}}g(p)=0$ then $g\otimes \chi \in D(A^i_n)$ for $\chi\in\cF_0$. Moreover, it is clear that choosing $g$ with a suitable $\tilde{\g}$, we can also construct vectors in the domain of any finite product of DDF operators. Moreover, taking into account that $\tilde{L}_m\chi=0$ for $\chi\in\cF_0$  and $m\in\bN$ sufficiently large and that there always exists a positive integer $\b$ (depending on $m$ and $\chi$) such that
\begin{equation}
\tilde{L}_{m}\chi = \frac{1}{2} \sum_{n \in \bZ/{0}, |n|\leq \b}:\alpha_{m-n}\cdot \alpha_{n}:\chi
\end{equation}
we see that we can choose $\tilde{\g}$ so that the constraints equations make sense toghether with the commutation relations (\ref{ddf2}).
We are now ready for the following
\begin{proposition}\label{compact}
Fix $\chi\in\cF_0$ and any finite product of transverse DDF operators $A^{i_1}_{-n_{1}}\dots A^{i_k}_{-n_{k}}$ ($n_1,\cdots, n_k >0$). Then there are functions $g$ such that:
\begin{enumerate}
\item\label{item0}
there holds $g\otimes\chi\in D(A^{i_1}_{-n_{1}}\dots A^{i_k}_{-n_{k}})$
\item\label{item1}
the vectors $g\otimes\chi$ and $\hat{\psi}=(A^{i_1}_{-n_{1}}\dots A^{i_k}_{-n_{k}})(g\otimes\chi)$ are in $D(L_m)$ for all $m\in\bN$.
\item\label{item2}
the vector $\hat{\psi}$ can be written as a finite sum of terms of the type $g_i\otimes\chi_i$, where $\chi_i\in\cF_0$ and each $g_i$ is the restriction of an entire analytic function $g_i(\xi)$ ($\xi\in\bC^d$) satisfying the bounds of the Paley-Wiener theorem with one and the same $R$ (see below). Thus $\hat{\psi}$ is the Fourier transform of a function $F\in C_0^{\infty}(\bR^d,\cF)$, the space of smooth functions with compact support from $\bR^d$ to $\cF$.
\end{enumerate}
\end{proposition}
\begin{proof}
Choose a entire analytic function $f(\xi)$ ($\xi\in\bC^d$) such that for each $N\in\bN$ there exist positive constants $C_N$ and $R$ such that
\begin{equation}
|f(\xi)|\leq\frac{C_N e^{R|\Im\xi |}}{(1+|\xi |)^N}
\end{equation}
for all $\xi\in\bC^d$ (the Paley-Wiener bounds, see \cite{Reed1980}). Then from the proof of theorem \ref{thm:ddf} and the previous discussion we see there is a positive integer $\g$ such that the entire analytic function $g(\xi)=(\xi^0+\xi^d)^{\g}f(\xi)$ satisfies $g\otimes\chi\in D(A^{i_1}_{-m_{1}}\dots A^{i_k}_{-m_{k}})$ and the conditions in items \ref{item0} and \ref{item1}. Moreover, $\hat{\psi}$ is a finite sum of terms of the type
\begin{multline}\nonumber
K\cdot [(p^{i_{1}}\cdots p^{i_{l}})\frac{1}{(p^0+p^{d-1})^{\g^\prime}} f(p)]\otimes [(\alpha_{-n_{1}}^0 +\alpha_{-n_{1}}^d) \cdots\\
\cdots(\alpha_{-n_{q}}^0 +\alpha_{-n_{q}}^d)\times (\alpha_{n^\prime_{1}}^{0}+\alpha_{n^\prime_{1}}^d) \cdots (\alpha_{n_{s}^\prime}^0 +\alpha_{n_{s}^\prime}^d)\chi]
\end{multline}
where once more $l,\g^\prime \in\bN$ and $K\in\bC$ are constants depending on $\chi\in\cF_0$ and the $A^j_n$'s. Since each of the finite number of functions $(p^{i_{1}}\cdots p^{i_{l}})(p^0+p^{d-1})^{-\g^\prime} g(p)$ by construction still satisfies the Paley-Wiener bounds with different constants $C_N$ but {\it one and the same} $R$, item \ref{item2} is proved.
\end{proof}
\bigskip
\begin{Remark}\begin{enumerate}
\item
We observe that the constant $R$ can be chosen arbitralily small. Thus we obtain vectors $\psi\in C_0^{\infty}(\bR^d,\cF)$ with support contained in $d$-spheres with arbitrary radius.
\item
The states we constructed are clearly in $\sS(\bR^d)\otimes\cF_0$ and as a consequence in the domain of the operators $U(a,\L)$ implementing the representation of the Poincar\'e group.
\end{enumerate}\end{Remark}

\section{Free string field theory}\label{field}

After all this preparation we come to the main subject of this work, namely the study of the localization properties of {\it second} quantized open bosonic strings. To this aim, we adopt Dimock's approach to the construction of free string field theory developed in \cite{Dimock2002}, which we briefly recall below. The string field turns out to be a more or less ordinary $local$ indefinite metric quantum field theory (see \cite{Strocchi1993}, but also \cite{Mintchev1980} for a more specific study of the Fock construction). The word ``local'' will mean as usual that the fields constructed below will have vanishing commutators for space-like separated arguments.\\
First, to obtain a casual theory we need to exlude tachyons from the one particle state space and thus we define
\begin{equation}\label{eq:pos}
\cH_+ =\sum_{r\geq0}\cK_r
\end{equation}
For any  $F(x) \in C^{\infty}_0(\bR^d ,\cF,dx)$ (the space-time representation) we introduce the projection  $\Pi F  \in \cH$ by specifying that its restriction to the hyperboloid $V^+_r$ is given by
\begin{equation}\label{Pi}
(\Pi F)_r(p)    =  \sqrt{ 2 \pi} \ P_r \hat F ( p )  
\end{equation}
where the hat denotes Fourier transformation and $P_r$ the projection on the subspace $M^2=r$. We note that $\Pi F \in \cH_+$ for   $  F \in \sS(\bR^d  ,  \cF)$.\\
We are now ready for the following
\begin{definition}
Let $a^\dagger$ and $a$ denote the usual creation and destruction operators (see \cite{Mintchev1980} but also \cite{Reed1975}) on the symmetric Krein-Fock space
\begin{equation}
\sH = \G(\cH_+)= \bigoplus_{n=0}^{\infty}(\cH_+^{(n)})_s
\end{equation}
(the subscript $s$ indicates symmetrization) with fundamental symmetry $\sJ=\G(J)$ and indefinite inner product $(\Psi,\sJ\Xi)$. Then the string field operator is defined as
\begin{equation}
\Phi (F)   = \frac{1}{\sqrt2} [ a^{\dagger} (\Pi F) + a (\Pi F )]
\end {equation}
on the domain $\sD_0\subset\sH$ of (quasilocal) finite particle vectors, {\it i.e.} if $\Psi\in\sD_0$ we have $\Psi=(\Psi^{(0)},\cdots,\Psi^{(n)},0,\cdots)$ where $\Psi^{(i)}\in \sS(\bR^d  ,  \cF)^{(i)}_s$.
\end{definition}
\bigskip

Combining \cite{Mintchev1980} and \cite{Dimock2002}, one has the following
\begin{theorem}  \label{stringfield}
\begin{enumerate}
\item The string field  $\Phi(F)$ is closable and $\sJ$-symmetric $\forall  F \in \sS(\bR^d  ,  \cF)$.
\item $\sD_0$ is an invariant dense set of analytic vectors for $\Phi (F)$.
\item If $\{F_k\}\subset\cH_+$ and $F_k\to F$ in the topology of $\sS(\bR^d  ,  \cF)$, then
\begin{equation}\nonumber
\text{s}-\lim_{k\to\infty}\Phi(F_k)\Psi=\Phi(F)\Psi
\end{equation}
\item The vacuum vector $\Omega=\{1,0,\cdots,0,\cdots\}\in\sD_0$ is cyclic.
\item For every $\Psi\in\sD_0$ and $F,G\in\sS(\bR^d  ,  \cF)$ we have
\begin{equation}\label{fieldcomm}
\Phi(F)\Phi(G)\Psi-\Phi(G)\Phi(F)\Psi=i\Im \langle\Pi F,\Pi G\rangle\Psi=-i\langle F,E G\rangle\Psi
\end{equation}
and the field equation $\Phi((-\square+M^2)F)\Psi=0$ holds. In particular, if  $F,G$ have spacelike separated supports
there holds (locality)
\begin{equation}\label{vanishing}
[ \Phi (F), \Phi(G)]\Psi  = 0.
\end{equation}
\item   There is a positive energy representation $\sU(a, \L)=\G(U(a, \L))$ (see definition \ref{singlestring}) 
of the Poincar\'e group on $\sH$ satisfying, for every $\Psi\in\sD_0$ such that $\sU(a, \L)^{-1}\Psi$ is defined and every $F\in D(u)$ (see \ref{poi}), the equation
\begin{equation}
\sU(a, \L)\Phi(F) \sU(a, \L)^{-1}\Psi
= \Phi (F_{a, \L})\Psi
\end{equation}
where  $F_{a, \L}(x)  =  u(\L) F(\L^{-1}(x - a))$.
\end{enumerate}
\end{theorem}
\bigskip The symbol $E$ in equation (\ref{fieldcomm}) indicates the propagator $E=E^+ - E^-$, with
\begin{equation}
(E^{\pm} F)(x)  = \frac{1}{2(2 \pi)^{d/2}}  \int_{\gamma^{\pm} \times \bR^{d-1}}
 \frac{ e^{ip\cdot x}}{p^2 + M^2} \ 
\hat  F(p) dp 
\end{equation}
for $F \in C^{\infty}_0(\bR^d ,\cF)$.  The  $p^0$ contour $\gamma^{\pm} $ in the integral is the real line shifted by a small imaginary quantity $\pm i\epsilon$ ($\epsilon>0$), while $M^2$ is of course the (square) mass operator defined in Lemma \ref{lem:mass}. Locality follows from the properties of $E$.\\
As already anticipated, constraints are imposed in a Gupta-Bleuler form: taking the (anti)-Fourier transform $\check{L}_m$ of the operators $L_m$, the condition is
\begin{equation}\label{fieldconst}
(\check{L}_m\Phi_-)(F)\Psi\equiv a(\Pi\check{L}_{-m} F)\Psi=a(L_{-m}\Pi F)\Psi=0 \quad\quad m\geq0
\end{equation}
for $\Psi\in\sH$. The final result is
\begin{equation}\begin{split}\nonumber
\sH^{phys}=\sK^{\prime}/\sK^{\prime\prime}\quad\quad\quad\quad\quad\\
\sK^{\prime}=\G(\cH_+^{\prime}) \quad\quad\quad\quad\quad \cH_+^{\prime}=\cH^{\prime}\cap\cH_+\\
\sK^{\prime\prime}=\sK^{\prime}\cap(\sK^{\prime})^{\perp}\quad\quad\quad\quad\quad
\end{split}\end{equation}
Finally, we have
\begin{equation}\nonumber
\sH^{phys}=\G(\cH_+^{phys})
\end{equation}
so that for $d=26$, $b=1$ the inner product restricts on $\sH^{phys}$ to a positive definite one.
As it is customary in Gupta-Bleuler field theory, the field $\Phi$ determines an operator on $\sH^{phys}$ whenever smeared on functions satisfying certain conditions. This leads to the following
\begin{definition}\label{testcon}
A test function $F\in\sS(\bR^d,\cF)$ is called constrained if $\Pi F\in \cH_+^{\prime}$.
\end{definition}
\bigskip
As in \cite{Dimock2002}, to get real constrained test functions it is convenient to choose the conjugation $C_1$ defined on $\sS(\bR^d)\otimes\cF_0$ by
\begin{equation}\begin{split}\nonumber
C_1 & \a_n^0C_1 =\a_n^0\\
C_1\a_n^iC_1=-&\a_n^0\quad\quad\quad\quad i=1,\cdots,d
\end{split}\end{equation}
Real will then mean $C_1\hat{F}=\hat{F}$.
We are finally ready for the following
\begin{theorem}
Constrained (real) test functions $F\in\sS(\bR^d,\cF_0)$ with support contained in $d$-spheres with arbitralily small radius $R$ exist. Moreover, the Poincar\'e transformations $F\to F_{a, \L}$ (see Theorem \ref{stringfield}) are well defined for such $F$'s.
\end{theorem}
\begin{proof}
Consider any function $G\in\sS(\bR^d,\cF_0)$ with Fourier transform
\begin{equation}\nonumber
\hat{G}=(A^{i_1}_{-n_{1}}\dots A^{i_k}_{-n_{k}})(f\otimes\Omega)
\end{equation}
as constructed in proposition \ref{compact} (but see also the remark thereafter). Of course here $\Omega\in\cF$ indicates the oscillator one particle space vacuum. From equations (\ref{constraintsf}) and (\ref{salvezza1}), (\ref{salvezza2}) and the very definition of $\Pi$ (see equation \ref{Pi}), we see that
\begin{equation}\begin{split}\nonumber
L_m\Pi F=L_m( \ocp,\bp) \hat{G}( \ocp,\bp)=0 \qquad m>0\\
\nonumber L_0\Pi G=L_0( \ocp,\bp)\hat{G}( \ocp,\bp)=0\qquad\quad
\end{split}\end{equation} (with $c=2(-b+\bar{n})$ and $\bar{n}=\sum_{l=1}^{k}n_l$) so that $\Pi G\in \cH_+^{\prime}$. Moreover, since $C_1 L_m( \onp, \bp) C_1 =   L_m( \onp,- \bp)$ the function
\begin{equation}\nonumber
\hat{G}^\prime( \onp,\bp)=\hat{G}( \onp,\bp)+C_1\hat{G}( \onp,- \bp)
\end{equation}
is real and satisfies both the Paley-Wiener bounds with the same $R$ and the constraints. Finally, since $\hat{G}\in\sS(\bR^d,\cF_0)$ the transformations $G\to G_{a, \L}$ are well defined as claimed. Fixing the support and taking (finite real) linear combinations, we obtain a rich class of functions with the desired properties.
\end{proof}

\emph{Acknowledgements.} The authore would like to thank Robbert Dijkgraaf for several interesting and useful discussions during his stay in Amsterdam for his Phd thesis and Sergio Doplicher for his constant support and encouragement.
\providecommand{\bysame}{\leavevmode\hbox to3em{\hrulefill}\thinspace}

\end{document}